\documentclass[a4paper]{article}
\usepackage{array}
\usepackage{float}

\usepackage{amsmath} 
\usepackage{amssymb}  
\usepackage{theorem}

\def\Q{\mathcal{Q}}
\def\S{\mathcal{S}}
\def\md{\mathbf{d}}
\def\B{\mathcal{B}}
\def\E{\mathcal{E}}

\def\e{\epsilon}
\def\d{\delta}
\def\ab{($\e$,\,$\d$)}
\def\Prob{\mathbb{P}}

\def\dx{\text{\, dx}}

\DeclareMathOperator{\diam}{diam}
\DeclareMathOperator{\dist}{dist}

\newcommand{\A}{\mathcal{A}}

\newtheorem{lemma}{Lemma}
\newtheorem{theorem}{Theorem}
\newtheorem{corollary}{Corollary}
\newtheorem{definition}{Definition}

\newtheorem{proposition}{Proposition}
\newtheorem{example}{Example}

\begin{document}

\title{Differential Privacy in Metric Spaces: Numerical, Categorical and Functional Data Under the One Roof}
\author{Naoise Holohan \thanks{Hamilton Institute, NUI Maynooth, Co. Kildare, Ireland}, Douglas Leith \thanks{Hamilton Institute, NUI Maynooth, Co. Kildare, Ireland} Oliver Mason\thanks{Hamilton Institute, NUI Maynooth, Co. Kildare, Ireland. Corresponding Author: email: oliver.mason@nuim.ie.  Supported by the HEA, PRTLI 4 Network Mathematics Grant}}

\maketitle

\begin{abstract}
We study Differential Privacy in the abstract setting of Probability on metric spaces.  Numerical, categorical and functional data can be handled in a uniform manner in this setting.  We demonstrate how mechanisms based on data sanitisation and those that  rely on adding noise to query responses fit within this framework.  We prove that once the sanitisation is differentially private, then so is the query response for any query.  We show how to construct sanitisations for high-dimensional databases using simple 1-dimensional mechanisms.  We also provide lower bounds on the expected error for differentially private sanitisations in the general metric space setting.  Finally, we consider the question of sufficient sets for differential privacy and show that for relaxed differential privacy, any algebra generating the Borel $\sigma$-algebra is a sufficient set for relaxed differential privacy.  
\end{abstract}


\section{Introduction}

\subsection{Background}

The rapid expansion of the Internet and its use in everyday life, alongside the growing understanding of the potential benefits of big data \cite{bigdatause}, has pushed data privacy to the forefront of research priorities since the turn of the millennium.
Whether it be online, in the supermarket or at the hospital, corporations and governments are collecting vast quantities of data about our activities, the choices we make and the people we are in order to work more efficiently, increase profits and better serve our needs as consumers and citizens \cite{CSC}.
The challenge of making this potentially highly sensitive data publicly available where it can be put to good use is far from trivial and it is with this problem that the field of data privacy is concerned.

Various researchers and practitioners have considered applying anonymisation techniques to data sets such as removing explicit identifiers (name, address, telephone number, social security number, etc) while leaving quasi-identifiers\footnote{A quasi-identifier is an attribute that is not sufficient to identify an individual by itself, but can do so when combined with other quasi-identifiers ({\em e.g.} gender, date of birth, etc.).} in place.
While these anonymised data sets do indeed preserve participants' privacy in isolation, auxiliary/background information make this technique extremely vulnerable to attack \cite{Machan}.  A study by L. Sweeney in 2000 \cite{Sweeney} found that as much as 87\% of the US population (216 out of 248 million people) could be uniquely identified using only three quasi-identifiers (5-digit ZIP code, gender and date of birth).
This meant census data could be linked to ``anonymised'' health records to determine the health status of unsuspecting patients.

Then in 2006, American media firm AOL released 20 million Internet search queries, with user numbers in place of other quasi-identifiers to protect users' identity. This shield of anonymity was not sufficient for privacy to be protected however, and the data was quickly removed from the public domain \cite{AOL}.  Similarly in 2008, A. Narayanan and V. Shmatikov \cite{Narayanan} successfully de-anonymised entries in an anonymised data set containing movie ratings of 500,000 subscribers which was released by movie streaming website Netflix.
The authors used the publicly-available Internet Movie Database as background information and were able to positively identify known users despite the absence of explicit identifiers in the Netflix data set.  Initial anonymisation methods such as $k$-anonymity \cite{Swe02} and $l$-diversity \cite{Machan} have been shown to be vulnerable to attacks based on background information \cite{Machan, Li}.  This fact has led to the development of a wide variety of more sophisticated approaches to publishing and mining data anonymously \cite{FungBOOK, tClose, FerrBOOK}.

The work discussed above underlines the unsatisfactory nature of ad-hoc privacy solutions and the need for a solid theoretical foundation for privacy research.  With this in mind, the concept of differential privacy was proposed in \cite{Dwork1} to provide a formal, mathematical framework for analysing privacy-preserving data publishing and mining.  The premise of differential privacy is that the outputs of queries to a database are unlikely to change substantially with the addition of a new participant's information.  This means that outputs will be similar whether or not an individual participates in the database.

There is now a considerable body of work on differential privacy in the theoretical Computer Science literature \cite{DworkSurvey}.  Many of the paper in the literature concern data or queries of some particular type or on the development of particular algorithms that satisfy differential privacy.  For instance, the design of differentially private algorithms for calculating singular vectors is considered in \cite{PrivSVD2}, while differentially private recommender systems are developed in \cite{SHERRY}; in both of these instance, the data are naturally modelled as real numbers.  Algorithms for search problems and learning are considered in \cite{SEARCH, LEARNING}.  A statistical perspective on differential privacy was developed in \cite{WassStat}; this paper considered real-valued ($[0,1]$ in fact) queries and data.  In \cite{ROUGHG}, mechanisms that maximise a suitable utility function were investigated; this paper assumed discrete finite-valued data spaces, which can of course describe categorical data.  The recent paper \cite{DOMFERR} addressed the problem of optimal mechanisms that add noise independent of the data; the queries considered are real-valued.  To date, the only major reference on differential privacy for functional data appears to be \cite{WassFunc}; in the same paper the authors emphasise the importance of being careful in selecting the measure space with respect to which probabilities are defined.  In particular, if we choose our $\sigma$-algebra to be the trivial one consisting of the empty set and the entire space, then every mechanism is differentially private.  The work of \cite{HarTal} and other similar papers on lower bounds for differentially private mechanisms considers real (or in some cases integer) valued data.  Our aim is to describe a unifying framework for all major data and query types considered so far and to initiate a study of differential privacy in the formal setting of probability measures on metric spaces.  This provides the machinery necessary for a rigorous discussion of random mechanisms and their accuracy.  The major aims of the paper are to place the study of differentially private mechanisms within the framework of probability measures on metric spaces and to present some initial results in this direction.

\subsection{Our Results}
The principal contributions of this paper are the following.  
\begin{itemize}
\item We consider differential privacy in the general framework of probability on metric spaces and highlight that it can be seamlessly applied to numerical, categorical and functional data.  
\item Our description shows how mechanisms based on database sanitisation and adding noise to query responses can be treated in a unified fashion.
\item We describe techniques for generating families of {\ab} differentially private mechanisms from simpler mechanisms.  One such example is given by sanitisation mechanisms generated from an {\ab} differentially private mechanism for the identity query.  We also show how to generate differentially private mechanisms for high-dimensional databases using mechanisms for 1-dimensional databases.
\item We describe lower bounds for the error in releasing a database in an {\ab}-differentially private fashion using product sanitisations.  This result applies to data drawn from any metric space in contrast to previous work, which has largely focussed on real and integer valued data.
\item We consider the question of testing differential privacy and describe sufficient sets for {\ab} differential privacy.
\end{itemize}

\subsection{Structure of paper}
We begin in Section~\ref{sec:prelim} by establishing the measure-theoretic framework for differential privacy.  We then consider the question of sufficient sets for differential privacy in Section \ref{sec:Suff}, address sanitisation mechanisms in Section \ref{sec:San} and focus on product sanitisations in Section \ref{sec:Prod}.  Section~\ref{sec:Acc} considers accuracy and we give concluding remarks in Section \ref{sec:Conc}.

\section{Preliminaries} \label{sec:prelim}
We first recall some standard concepts and results from probability and measure theory \cite{Bill, Rudin}.  Given an algebra $\S$ of subsets of a set $\Omega$, we use $\sigma(\S)$ to denote the smallest $\sigma$-algebra containing $\S$ and refer to $\sigma(\S)$ as the $\sigma$-algebra generated by $\S$.  A set $\Omega$ together with a $\sigma$-algebra of subsets of $\Omega$ is a measurable space.  

A monotone class $\mathcal{M}$ of subsets of some set $\Omega$ is defined by the following two properties: (i) if $\{A_i\}_{i=1}^\infty \subseteq \mathcal{M}$, and if $A_i\subseteq A_{i+1}$ for all $i$, then $\bigcup_{i=1}^\infty A_i\in\mathcal{M}$; (ii) if $\{A_i\}_{i=1}^\infty \subseteq \mathcal{M}$, and if $A_i\supseteq A_{i+1}$ for all $i$, then $\bigcap_{i=1}^\infty A_i\in\mathcal{M}$.

The next result, which appears as Theorem~3.4 in \cite{Bill}, characterises $\sigma(\S)$ as the smallest monotone class containing $\S$. 
\begin{theorem}
\label{thm:MCT} Let $\S$ be an algebra of subsets of some set $\Omega$ and let $\mathcal{M}$ be a monotone class such that $\S \subseteq \mathcal{M}$. Then $\sigma(\S) \subseteq \mathcal{M}$.
\end{theorem}

Given two measurable spaces $(X, \A_X)$ and $(Y, \A_Y)$, subsets of $X\times Y$ of the form 
$$R = \bigcup_{i=1}^p X_i \times Y_i,$$
where $X_i \in \A_X$, $Y_i \in \A_Y$ for $1 \leq i \leq p$ and $(X_i \times Y_i) \cap (X_j \times Y_j) = \emptyset$ for $i \neq j$ are known as \emph{elementary subsets}.  Let $\mathcal{R}$ denote the collection of all elementary subsets and denote the usual product $\sigma$-algebra on $X \times Y$ by $\A_{X \times Y}$.  The following result is Theorem~8.3 of \cite{Rudin}.
\begin{theorem}
\label{thm:Rudin} If $\mathcal{M}$ is a monotone class and $\mathcal{R} \subseteq \mathcal{M}$, then $\A_{X \times Y} \subseteq \mathcal{M}$.
\end{theorem}

Finally, for a measure $\mu$ on a measurable space $(X, \A_X)$, we recall the following simple fact.
\begin{proposition}\label{pro:Measure}
Suppose that $\{A_i\}_{i=1}^\infty \subseteq \A_X$ satisfies $A_i\subseteq A_{i+1}$ for all $i$, then $\lim_{i\to\infty} \mu(A_i)=\mu\left(\bigcup_{i=1}^\infty A_i\right)$.

Similarly, if $A_i\supseteq A_{i+1}$ for all $i$, then $\lim_{i\to\infty} \mu(A_i)=\mu\left(\bigcap_{i=1}^\infty A_i\right)$.
\end{proposition}

\subsection{Database Model}
The individual entries of the databases we consider are elements of a set $D \subseteq U$ where $U$ is a metric space with metric $\rho$.  We equip $U$ with the Borel $\sigma$-algebra generated by the open sets in $U$ (in the metric topology);  $D$ then naturally inherits a $\sigma$-algebra $\A_D$.  A database $\md$ with $n$ rows is given by a vector $\md=(d_1,\dots,d_n) \in  D^n$ in which $d_i \in D$ is the $i$th row.  Throughout, we assume that $U^n$ (and $D^n$) is equipped with the usual product $\sigma$-algebra $\A_{U^n}$ generated by $\{A_1\times\cdots\times A_n: A_i\in\A_U \}$.  This ensures that projection maps $\pi_i:U^n \rightarrow U$ given by $\pi_i(x_1, \ldots, x_n) = x_i$ are measurable.    

It is worth highlighting the generality of this setting: the metric space $D$ can contain numerical, categorical or functional data; moreover, it can be discrete or continuous.  

\begin{example}
\label{ex:sets}
If our data concern the hobbies or interests of people, we consider a set of all possible hobbies, denoted by $\mathcal{H}$.  For simplicity it is not unreasonable to assume that $\mathcal{H}$ is finite.  Our data entries are then drawn from the power set $D := 2^{\mathcal{H}}$ of $\mathcal{H}$, which will again be a finite set.  There are various natural choices of metric in this case.  We could consider the discrete metric on $D$ in which $\rho_1(A, B) = 1$ if $A \neq B$ and 0 otherwise.  Alternatively, we could choose the metric given by symmetric distance: $\rho_2(A, B) = |(A \cup B) \backslash (A \cap B)|$.  In both of these cases, the Borel $\sigma$-algebra consists of all subset of $D$.   Note that there is no requirement that each entry in a database in $D^n$ have the same size or cardinality, reflecting the fact that not all of us have the same number of interests or hobbies.    
\end{example}

\begin{example}
\label{eq:funcs}
In readings from field deployed sensors, each reading has a time-stamp giving rise to time-course data.  Another example is in smart metering where the supplier collects data from consumers giving electriciy consumption over a time-window.  Data of this type is naturally represented as either a function or a sequence of real numbers.  In our framework, we can take $U$ to be a sequence space such as $l_{\infty}$ or $l_2$, or an appropriate function space such as $C([0, T])$ or $L_2([0,T])$, where $T$ represents the billing period (for instance).   All of these spaces have natural norms defined on them (in fact they are all Banach spaces) and can be equipped with the Borel $\sigma$-algebra generated from the open sets in the norm topology.  
\end{example}

We say that two databases $\md=(d_1,\dots,d_n)$ and $\md^\prime=(d_1^\prime,\dots,d_n^\prime)$ in $D^n$ are \emph{neighbours}, and write  $\md \sim \md^\prime$, if there is some $j \in  \{1,\dots,n\}$ such that $d_j\neq d_j^\prime$ and $d_i = d_i^\prime$ for all $i\in\{1,\dots,n\}\setminus\{j\}$.  More generally, we denote by $h(\md, \md')$ the \emph{Hamming distance} between $\md$ and $\md'$. 

\textbf{Remark:}  While we present our results for Hamming distance on $D^n$, the results of Sections \ref{sec:Suff} and \ref{sec:San} are also valid with respect to other metrics $\rho$ on $D^n$ where we define $\md \sim \md^\prime$ if $\rho(\md, \md') = 1$.  The work of Sections \ref{sec:Prod} and \ref{sec:Acc} relies on similarity being defined with respect to Hamming distance however.

For the most part, we assume that the data space $D$ is compact.  This is immediate if $D$ is finite (as in Example \ref{ex:sets}) and is a natural assumption in most realistic situations.  
When $D$ is compact, we denote by $\diam(D)$ the diameter of $D$:
\begin{align}\label{eq:diam}
\diam(D) = \max_{d, d^\prime \in D} \rho(d, d').
\end{align}

\subsection{Query Model}
We consider a very general query model.  The set of all possible responses is assumed to be a metric space $E_Q$ with metric $\rho_Q$ and equipped with the Borel $\sigma$-algebra $\A_Q$.   The query is then a measurable function, $Q: U^n \to E_Q$ and hence $Q^{-1}(A)\in\A_{U^n}$ for all $A \in \A_Q$.  

\begin{example}
As with the data in $\md$, queries are not restricted to take numerical values in this setting.  For instance, if we consider Example \ref{ex:sets} above, then we could consider a query asking for the number of people in the database who are interested in Classical Music or Football for instance: this would clearly be a numeric query.  On the other hand, we could also request the 3 most common hobbies in the database, the output of which would be a set.  
\end{example}
We next formally introduce the concept of a \emph{response mechanism} within this general framework.  Let $(\Omega, \mathcal{F}, \mathbb{P})$ be a probability space.  Given a collection of queries $\mathcal{Q}(n)$, a response mechanism is a set of measurable mappings
\begin{equation}\label{eq:Mech1}
\{X_{Q, \md}: \Omega \to E_Q \mid Q \in \Q(n)\}.
\end{equation}
Note that $X_{Q, \md}$ is an $E_Q$-valued random variable for each $Q$ and $\md$. 

\textit{Sanitised Response Mechanisms}

If there exists some family $$\{X_{\md}: \Omega \to U^n\}$$ of measurable mappings such that 
\begin{equation}\label{eq:Mech2}
X_{Q, \md} = Q \circ X_{\md}
\end{equation}
for all $\md \in D^n$, then the mechanism given by (\ref{eq:Mech2}) is a \emph{sanitised response mechanism}.  The motivation behind this choice of terminology is that the mechanism is generated by sanitising the database via the random variable $X_{\md}$ before answering the query $Q$, \emph{i.e.} non-interactive data release.  If the database is sanitised by adding appropriate noise, the mapping $X_{\md}$ takes the form $X_{\md} = \md + N$ for some $U^n$-valued random vector $N$.  Clearly, in order to define a mechanism by adding noise, it is necessary for $U^n$ to have a suitable algebraic structure: a vector space or monoid for example.  

\textit{Output Perturbations}

Mechanisms that provide Differential Privacy by perturbing the query response can be described as follows.  Let a query $Q:U^n \rightarrow E_Q$ be given.  Assume that there is a family of measurable functions $\{ X_q: \Omega \rightarrow E_Q \mid q \in E_Q\}$.  The output perturbation mechanism is defined as 
\begin{equation}\label{eq:PertMech}
X_{Q, \md} =  X_{Q(\md)}
\end{equation}
In situations involving real-valued data, $X_q$ typically takes the form $X_q = q + N$ where $N$ represents the noise added to the query response.  For set-valued queries, $X_q$ can be defined by specifying some probability mass function on the collection of possible query responses.      
\subsection{Differential Privacy}
In the interest of completeness and clarity, we now recall the definition of Differential Privacy and write it in the setting of this paper. 
\begin{definition}[Differential Privacy with Respect to a Query]
Let $\e \ge 0$, $0\le \d \le 1$ be given.  A response mechanism is {\ab}-differentially private with respect to a query $Q_0\in\mathcal{Q}(n)$ if for all $\md \sim \md^\prime\in D^n$ and all $A \in \A_{Q_0}$,
\begin{equation}\label{eq:DP}
\mathbb{P}(X_{Q_0,\md} \in A) \le e^\e \mathbb{P}(X_{Q_0, \md^\prime} \in A)+\d.
\end{equation}
It is important to note that the relation $\md \sim \md^\prime$ is symmetric, so inequality (\ref{eq:DP}) is required to hold when $\md$ and $\md^\prime$ are swapped.
\end{definition}

\begin{definition}[Differential Privacy]
A response mechanism is {\ab}-differentially private with respect to $\mathcal{Q}(n)$ if it is {\ab}-differentially private with respect to every query $Q_0\in \mathcal{Q}(n)$.
\end{definition}
The above definitions are often referred to as \emph{relaxed} differential privacy; the original notion of differential privacy introduced in \cite{Dwork1} considers the case where $\d = 0$.  

The \emph{expected error} of a mechanism (\ref{eq:Mech1}) for a query $Q$ on a database $\md$ is given by the expectation $\mathbb{E}[\dist(X_{Q, \md}, Q(\md))]$.  

\section{Sufficient Sets for Differential Privacy} \label{sec:Suff}
In this brief section, we consider the following question: is there a strict subset $\S$ of $\A_Q$ such that if (\ref{eq:DP}) is satisfied for all $A$ in $\S$, it is guaranteed to be satisfied for all $A$ in $\A_Q$?  We refer to such a set as a sufficient set for differential privacy.


Depending on the application, the query output space may be a subset of $\mathbb{R}^n$ or of a sequence or function space such as $C([0, T])$.  A key question for the practical deployment of differentially private mechanisms is how to determine if a mechanism is in fact {\ab}-differentially private.  Testing (\ref{eq:DP}) on the entire $\sigma$-algebra is clearly a prohibitively difficult task.  Our next result shows that it is sufficient to test this condition on any \emph{algebra} $\S \subset \A_Q$ that generates $\A_Q$.  

\begin{theorem}
\label{thm:Suff} Let a response mechanism (\ref{eq:Mech1}) and a query $(E_Q, \A_Q, Q)$ be given and let $\S \subset \A_Q$ be an algebra such that $\sigma(S) = \A_Q$.  If (\ref{eq:DP}) holds for all sets $A \in \S$, then it holds for all sets $A \in \A_Q$.
\end{theorem}
\begin{proof}
Let $\B$ denote the collection of sets in $E_Q$ for which (\ref{eq:DP}) holds.  By assumption, $\S \subseteq \B$.  Now let $A_1, A_2, \ldots$ be any collection of sets in $\B$ with $A_{i} \subseteq A_{i+1}$ for all $i$.  Define $\bar{A} := \cup_i A_i$ and let $\md, \md^\prime \in D^n$ with $\md \sim \md'$ be given.   As each $A_i \in \B$, it follows that 
$$\mathbb{P}(X_{Q, \md}(\omega) \in A_i) \leq e^{\e} \mathbb{P}(X_{Q, \md'}(\omega) \in A_i) + \d$$
for all $i$.  As the sequence $A_i$ is increasing, it now follows from Proposition \ref{pro:Measure} that 
\begin{align*}
\mathbb{P}(X_{Q, \md}(\omega) \in \bar{A})&= \lim_{i\to\infty} \mathbb{P}(X_{Q, \md}(\omega) \in A_i) \\
& \le e^\e \lim_{i\to\infty} \mathbb{P}(X_{Q, \md'}(\omega) \in A_i)+\d \\
&= e^\e \mathbb{P}(X_{Q, \md'}(\omega) \in \bar{A})+\d,
\end{align*}
and so $\bar{A}\in\B$.
An identical argument shows that for any sequence $\{A_i\}$ of sets in $\B$ with $A_{i} \supseteq A_{i+1}$ for all $i$, and $\bar{A} = \cap_i A_i$, $\bar{A} \in \B$.  Taken together these two observations imply that $\B$ is a \emph{monotone class}.  Moreover, $\S \subseteq \B$.  The result now follows immediately from Theorem~\ref{thm:MCT}.
\end{proof}

In the next example we show how Theorem \ref{thm:Suff} can be applied to differentially private mechanisms for functional data to obtain results such as those described in Section 3.1 of \cite{WassFunc}.  
\begin{example}\label{ex:Func}
Suppose our query $Q$ takes values in the space $C([0, 1])$ of continuous functions on $[0,1]$ equipped with the norm $\|f\|_{\infty} = \textrm{sup}\{|f(t)| : t \in [0,1]\}$ and the $\sigma$-algebra $\mathcal{A}_Q$ of Borel sets generated by the norm topology.  Let a mechanism $X_{Q, \md}$ be given.  Then $X_{Q, \md}(\omega)$ is in $C([0,1])$ for each $\omega \in \Omega$.  

Given a positive integer $k$ and real numbers $0 \leq t_1 < \cdots < t_k \leq 1$, consider the projection $\pi_{t_1, \ldots , t_k}: C([0,1]) \rightarrow \mathbb{R}^k$ given by 
$$\pi_{t_1, \ldots, t_k}(f) = (f(t_1), \ldots, f(t_k)).$$
These mappings are measurable with respect to the usual Borel $\sigma$-algebra on $\mathbb{R}^k$ and hence we can define the $\mathbb{R}^k$-valued mechanism $X^{t_1, \ldots t_k}_{Q, \md} = \pi_{t_1, \ldots , t_k} \circ X_{Q, \md}.$

We claim that if the finite-dimensional mechanisms $X^{t_1, \ldots t_k}_{Q, \md}$ are {\ab}-differentially private for all $k$ and $t_1, \ldots, t_k$, then the mechanism $X_{Q, \md}$ is {\ab}-differentially private.  The argument to show this is as follows.  From the assumption on the finite-dimensional mechanisms, it follows immediately that (\ref{eq:DP}) holds for all (so-called cylinder sets) sets $A$ of the form 
$$A = \pi^{-1}_{t_1, \ldots, t_k}(B)$$
where $B$ is a Borel set in $\mathbb{R}^k$.  These sets form an algebra and it follows from Theorem VII.2.1 (page 212) of \cite{Parth} that the $\sigma$-algebra they generate is the Borel $\sigma$-algebra of $C([0,1])$.  It follows immediately from Theorem \ref{thm:Suff} that $X_{Q, \md}$ defines an {\ab}-differentially private mechanism on $C([0,1])$ as claimed. 
\end{example}

\section{Sanitisation Mechanisms and the Identity Query} \label{sec:San}
A popular approach to designing differentially private response mechanisms is to add \emph{noise} to the query response $Q(\md)$.  It is known however, that this can lead to privacy compromises by averaging a large number of responses to an identical query \cite{DworkSurvey}, unless the number or type of queries that can be asked is restricted.  We now show that if a sanitisation mechanism (\ref{eq:Mech2}) is {\ab}-differentially private with respect to the identity query, then it is {\ab}-differentially private with respect to \emph{any} query.  It is important to appreciate that we place no restrictions on the query $Q$; it could map into a sequence space such as $l_{\infty}$ and represent an infinite sequence of individual queries for example. 

The identity query is defined by the identity map on the ambient space $U^n$.  
In sanitised response mechanisms, the ``perturbed database'' $X_\md(\omega)$ is disclosed by the response to the identity query.  However, disclosing $X_\md(\omega)$ need \emph{not} disclose the original database $\md$ provided that an appropriate privacy-preserving perturbation has been applied.   Importantly, if {\ab}-differential privacy is achieved with respect to query $I_n$, then the response to \emph{any} query is {\ab}-differentially private.

\begin{theorem}[Identity Query]\label{th:rank2}
Consider a sanitised response mechanism as defined in (\ref{eq:Mech2}).  Suppose this response mechanism is {\ab}-differentially private with respect to the identity query. Then it is {\ab}-differentially private with respect to any query $(E_Q, \A_Q, Q)$.
\end{theorem}

\begin{proof}
Let $\md, \md'$ be two elements of $D^n$ with $\md \sim \md'$.  By assumption,
\begin{equation}
\label{eq:ID1}
\mathbb{P}(X_{\md}(\omega) \in E) \le e^\e \mathbb{P}(X_{\md'}(\omega) \in E) + \d,
\end{equation}
for all $E \in \A_{U^n}$.  
Let a query $(E_Q, Q, \A_Q)$ and $A \in \A_Q$ be given.  As $Q$ is measurable, $Q^{-1}(A) \in \A_{U^n}$.  Then, using (\ref{eq:ID1}),
\begin{align*}
\mathbb{P}(X_{Q, \md}(\omega) \in A) &= \mathbb{P}(Q(X_{\md}(\omega)) \in A) \\
&= \mathbb{P}(X_{\md}(\omega) \in Q^{-1}(A)) \\
& \le e^\e \mathbb{P}(X_{\md'}(\omega) \in Q^{-1}(A)) + \d \\
&= \mathbb{P}(Q(X_{\md'}(\omega)) \in A) + \d.
\end{align*}
Thus, the mechanism satisfies {\ab}-differential privacy with respect to $Q$ also.
\end{proof}

\begin{corollary}
Consider a sanitised response mechanism as defined in (\ref{eq:Mech2}) and suppose $I_n\in\Q(n)$. Then this response mechanism is {\ab}-differentially private with respect to $I_n$ if and only if it is {\ab}-differentially private with respect to every query $Q\in\Q(n)$.
\end{corollary}

\begin{proof}
``$\Rightarrow$'': Theorem~\ref{th:rank2}.

``$\Leftarrow$'': The response mechanism is {\ab}-differentially private with respect to every query $Q\in\Q(n)$ by assumption, therefore it must be {\ab}-differentially private with respect to the identity query $I_n$, since $I_n\in\Q(n)$.
\end{proof}

Observe that the number and details of queries $Q\in\mathcal{Q}(n)$ do not need to be specified in advance for Theorem~\ref{th:rank2} to hold, and so queries may be interactive and unlimited in number.   This highlights a fundamental difference between privacy mechanisms that perturb the query response (\emph{e.g.} by adding Laplacian or Gaussian noise) vs privacy mechanisms that perturb the database itself.  Namely, in the former the added noise can be averaged out by an adversary repeating a query multiple times, thereby requiring a limit to be placed on the number of queries allowed, while in the latter an averaging attack of this sort is impossible; a repeated query will simply receive the same answer each time. 

\section{Product Sanitisations} \label{sec:Prod}
In this section, we derive a result that relates differentially private sanitisation-based mechanisms for $n$-dimensional databases in $D^n$ to mechanisms for simple $1$-dimensional databases.  

Before showing how differentially private sanitisation mechanisms for databases in $D^n$ can be constructed from simple mechanisms for databases in $D$, we first establish a number of technical results.

\begin{lemma}\label{lm:decomp}
Let $A_1, \ldots, A_p$, $B_1, \ldots , B_p$ be two collections of non-empty sets.  Then the finite union $\bigcup_{i=1}^p (A_i \times B_i)$ can be written as
\[
\bigcup_{i=1}^p (A_i \times B_i) = \bigcup_{I \subseteq \{1, \ldots, p\}} (\widetilde{A}_I\times\widetilde{B}_I),
\]
where $\widetilde{A}_I = \bigcup_{i\in I}A_i$ and $\widetilde{B}_I = \bigcap_{i\in I} B_i \setminus \bigcup_{i\notin I} B_i$.  Moreover,   $\widetilde{B}_I\cap\widetilde{B}_J = \emptyset$ for all $I \ne J$.
\end{lemma}

\begin{proof}
We need to prove equality and disjointness of the decomposition.
Let $(a,b)\in\bigcup_{i=1}^p (A_i\times B_i)$. Then there exists at least one $i^*$ such that $(a, b)\in A_{i^*}\times B_{i^*}$. Let $I_b:=\{i: b\in B_i\}\subseteq\{1, \dots, p\}$ (note $i^*\in I_b$). Then $b\in \bigcap_{i\in I_b}B_i$, but $b\notin B_j$ for any $j\notin I_b$, otherwise $j$ would be an element of $I_b$. Hence $b\in \bigcap_{i\in I_b}B_i \setminus \bigcup_{i\notin I_b}B_i$. Also $a\in\bigcup_{i\in I_b} A_i$ since $a\in A_{i^*}$. Hence $(a,b)\in\bigcup_{I\subseteq\{1, \dots, p\}} \left(\bigcup_{i\in I}A_i \times \bigcap_{i\in I} B_i \setminus \bigcup_{i\notin I} B_i \right)$.

Let $(a,b)\in\bigcup_{I\subseteq\{1, \dots, p\}}\left(\bigcup_{i\in I}A_i \times \bigcap_{i\in I} B_i \setminus \bigcup_{i\notin I} B_i \right)$. Then there exists at least one $I^*\subseteq\{1, \dots, p\}$ such that $(a,b) \in \bigcup_{i\in I^*}A_i \times \bigcap_{i\in I^*} B_i \setminus \bigcup_{i\notin I^*} B_i$. Hence $a\in A_i$ for at least one $i\in I^*$ and $b\in B_i$ for all $i\in I^*$ and so there exists at least one $i\in I^*$ such that $(a,b)\in A_i \times B_i$ and so $(a,b)\in\bigcup_{i=1}^p (A_i\times B_i)$.

Finally, we show that $\widetilde{B}_I\cap\widetilde{B}_J=\emptyset$ if $I \neq J$.
To see this, note that if $I \neq J$, then we can without loss of generality assume that there is some index $k \in I$ that is not in $J$.  
Then any $x \in \widetilde{B}_I$ must be in $B_k$.  However, as $k \in J^C$, it follows that $x \in \cup_{i \notin J} B_k$ and hence that $x \notin B_J$.  This shows that the intersection is empty as claimed.  
\end{proof}

For future use, we note that an analogous argument to that given above can be used to show the following.
\begin{lemma}\label{lem:decomp2} 
Let $A_1, \ldots, A_p$, $B_1, \ldots , B_p$ be two collections of non-empty sets.  Then the finite union $\bigcup_{i=1}^p (A_i \times B_i)$ can be written as
$$\bigcup_{j=1}^q (\widetilde{A}_j \times \widetilde{B}_j),$$
where $\widetilde{A}_i \cap \widetilde{A}_j = \emptyset$ for $i \neq j$.
\end{lemma}

For the remainder of this section we consider a special form for the database sanitisation $X_\md(\omega)$. Suppose a family $\{X_d:\Omega\to U \mid d\in D\}$ of measurable mappings is given.  Define the mechanism $X_\md$ for $\md = (d_1, \ldots, d_n)$ by
\begin{equation}\label{eq:rvdef}
X_\md(\omega) = \left(X_\md^1(\omega), \dots, X_\md^n(\omega)\right),
\end{equation}
where the $X_\md^i$ are independent and $X_\md^i$ has the same distribution as $X_{d_i}$, for all $\md\in d^n, i\in\{1, \dots, n\}$.

We first note the following lemma concerning such mechanisms.
\begin{lemma}
\label{lem:ProdNec} 
Let a family $\{X_d: \Omega \to U \mid d \in D\}$ of measurable mappings be given and let $X_{\md}$ be defined by (\ref{eq:rvdef}).  If $X_{\md}$ is {\ab}-differentially private then  
\[
\mathbb{P}(X_d\in A) \le e^\e \mathbb{P}(X_{d^\prime}\in A)+\d,
\]
for all $d, d^\prime\in D, A \in \A_U$.
\end{lemma}
\begin{proof}
Let $d, d^\prime$ in $D$ be given.  If $d = d^\prime$, the result is trivial.  If $d \neq d^\prime$, take 
$\md = (d, d_2, \ldots, d_n)$, $\md^\prime = (d^\prime, d_2,\ldots, d_n)$ for any choice of $d_2, \ldots, d_n$ in $D$.  As $X_{\md}$ is {\ab}-differentially private and the projection $\pi_1:U^n \rightarrow U$ onto the first coordinate is measurable, it follows that for $A \in \A_U$:
\begin{eqnarray*}
\mathbb{P}(X_d \in A) &=& \mathbb{P}(\pi_1(X_{\md}) \in A) \\
&=& \mathbb{P}(X_{\md} \in \pi_1^{-1}(A)) \\
&\leq& e^{\e}\mathbb{P}(X_{\md^\prime} \in \pi_1^{-1}(A)) + \d \\
&=& e^{\e}\mathbb{P}(X_{d^\prime} \in A) + \d.
\end{eqnarray*}
\end{proof}
We next note that the converse of this result also holds. 
\begin{theorem}\label{th:alg}
Consider a family $\{X_d: \Omega \to U \mid d \in D\}$ of measurable mappings and assume that 
\[
\mathbb{P}(X_d\in A) \le e^\e \mathbb{P}(X_{d^\prime}\in A)+\d,
\]
for all $d, d^\prime\in D, A \in \A_U$.
Let $X_\md$ be as defined in (\ref{eq:rvdef}). Then,
\[
\mathbb{P}(X_\md \in A)\le e^\e \mathbb{P}(X_{\md^\prime}\in A)+\d,
\]
for all $\md\sim\md^\prime \in D^n$ and all $A \in \A_{U^n}$.
\end{theorem}

\begin{proof}
We shall use induction on $n$.  By assumption, the result is true for $n=1$.  Let $n > 1$ be given and assume that the result is true for all $k \leq n-1$.  

Assume that $\md$ and $\md^\prime$ differ in the first element so $d_1 \neq d_1'$ but $d_j = d_j'$ for $j \neq 1$.  Let 
\begin{equation}\label{eq:Rect}
R = \bigcup_{i=1}^p (A_i \times B_i),
\end{equation}
where $A_i \in \A_U$, $B_i \in \A_{U^{n-1}}$, be given.  It follows from Lemma \ref{lm:decomp} that we can write
\begin{equation}\label{eq:algdec1} 
R = \bigcup_{i=1}^q (\widetilde{A}_i \times \widetilde{B}_i),
\end{equation}
where $\widetilde{A}_i \in \A_U$, $\widetilde{B}_i \in \A_{U_{n-1}}$ for $1 \le i \le q$ and $\widetilde{B}_i \cap \widetilde{B}_j = \emptyset$ for $i \neq j$.  
Then, using the fact that the sets $\widetilde{B}_i$ are disjoint and the independence of the components of $X_{\md}$, $X_{\md'}$,
\begin{align*}
&\mathbb{P}(X_\md\in R)\\
&\quad= \sum_{i=1}^q \mathbb{P}(X_\md \in \widetilde{A}_i \times \widetilde{B}_i) \\
&\quad= \sum_{i=1}^q \mathbb{P}(X_{d_1}\in \widetilde{A}_i) \, \mathbb{P}(X_{(d_2, \dots, d_n)}\in \widetilde{B}_i)\\
&\quad\le \sum_{i=1}^q \left(e^\e \mathbb{P}(X_{d_1^\prime}\in \widetilde{A}_i)+\d \right) \mathbb{P}(X_{(d^\prime_2, \dots, d^\prime_n)}\in \widetilde{B}_i)\\
&\quad=  \, e^\e \sum_{i=1}^q \mathbb{P}(X_{\md^\prime}\in \widetilde{A}_i\times \widetilde{B}_i) + \d \, \mathbb{P}\left(X_{(d^\prime_2, \dots, d^\prime_n)}\in \bigcup_{i=1}^q \widetilde{B}_i\right) \\
&\quad\le  \, e^\e \mathbb{P}(X_{\md^\prime} \in R) + \d.
\end{align*}
If $d_1 = d_1'$, then $(d_2, \ldots, d_n) \sim (d_2', \ldots, d_n')$ and we can use Lemma \ref{lem:decomp2} and the induction hypothesis to conclude that
\begin{align*}
&\mathbb{P}(X_\md\in R)\\
&\quad\le \sum_{i=1}^q \mathbb{P}(X_{d_1^\prime}\in \widetilde{A}_i) \left( e^\e \mathbb{P}(X_{(d^\prime_2, \dots, d^\prime_n)}\in \widetilde{B}_i)+\d \right)\\
&\quad=  \, e^\e \sum_{i=1}^q \mathbb{P}(X_{\md^\prime}\in \widetilde{A}_i\times \widetilde{B}_i) + \d \, \mathbb{P}\left(X_{d_1^\prime}\in \bigcup_{i=1}^q \widetilde{A}_i\right) \\
&\quad\le  \, e^\e \mathbb{P}(X_{\md^\prime} \in R) + \d.
\end{align*}

Thus for any set $R$ of the form (\ref{eq:Rect}), we can conclude that 
\begin{equation}\label{eq:DPelem}
\mathbb{P}(X_{\md} \in R) \le e^\e\mathbb{P}(X_{\md'} \in R) + \d.
\end{equation}
In particular, (\ref{eq:DPelem}) holds for all elementary sets $R \in \A_{D^n}$.  A similar argument to that used in the proof of Theorem~\ref{thm:Suff} shows that the collection of all sets satisfying (\ref{eq:DPelem}) is a monotone class; furthermore this collection of subsets contains the elementary sets.  The result now follows from Theorem~\ref{thm:Rudin}.   
\end{proof}

\textbf{Remark:}  In the following subsection, we describe some simple applications of Theorem~\ref{th:alg}.  It is worth noting that it applies to any database space $D \subseteq U$, and to \emph{discrete} spaces in particular.  For mechanisms of the form (\ref{eq:rvdef}), it simplifies the task of testing the mechanism for differential privacy considerably.  For instance, if $D$ is a finite set with $|D|$ elements, then it is only necessary to check (\ref{eq:DP}) for all $\binom{|D|}{2}$ pairs of elements of $D$ and all $2^{|D|}$ subsets of $D$ to ensure differential privacy on $D^n$.  In general, we would have $n\binom{|D|}{2}|D|^{n-1}$ pairs of neighbouring elements and $2^{|D|^n}$ subsets to worry about!        
\subsection{Examples}

\begin{example}\label{ex:Laplacian}
Inspired by a now standard approach to designing differentially private mechanisms, we first consider a sanitisation mechanism for real-valued databases in which Laplacian noise is added to each element of the database.  Recall that a Laplacian random variable $X:\Omega\to\mathbb{R}$ with mean zero and variance $2b^2$ has a probability density function (PDF) is given by
\[
f(x)=\frac{1}{2b}e^{-\frac{|x|}{b}}.
\]

Let $D\subset\mathbb{R}$ be bounded; for each $d \in D$, let $X_d(\omega) = d + L(\omega)$ where $L: \Omega \to \mathbb{R}$ is a Laplacian random variable with mean zero and variance $2b^2$ such that
\begin{align*}
b \ge \frac{\diam(D)}{\e-\log(1-\d)}.
\end{align*}
The resulting sanitised response mechanism corresponding to (\ref{eq:rvdef}) is {\ab}-differentially private for any database in $D^n$.

To see this, note that by Theorem~\ref{th:alg}, it is sufficient to show that 
\begin{align*}
\int_{A} \frac{e^{-\frac{|x-d|}{b}} }{2b}\dx &\le e^\e \int_{A} \frac{ e^{-\frac{|x-d^\prime|}{b}} }{2b}\dx   +\d,
\end{align*}
for all $d,d^\prime\in D, A \in \mathcal{B}(\mathbb{R})$, where $\mathcal{B}(\mathbb{R})$ denotes the Borel $\sigma$-algebra on the real line $\mathbb{R}$.  Using the triangle inequality, $|x-d^\prime| \le |x-d| + |\Delta|$ where $\Delta=d^\prime-d$, and so it is sufficient to show  that
\begin{align*}
\int_{A} \frac{e^{-\frac{|x-d|}{b}} }{2b}\dx &\le e^{\e- \frac{|d^\prime-d|}{b}}\int_{A} \frac{ e^{-\frac{|x-d|}{b}} }{2b}\dx +\d,
\end{align*}
for all $d^\prime,d\in D, A \in \mathcal{B}(\mathbb{R})$.  This last inequality will follow if $1 \le e^{\e- \frac{|\Delta|}{b}}+ \d$ or $b\ge \frac{|\Delta|}{\e-\log(1-\d)}$ for all $\Delta\in\{d^\prime-d:d^\prime, d\in D\}$.  

Of course, keeping in mind that the $l_1$ sensitivity of the identity query \cite{DworkSurvey} is precisely given by $\diam(D)$, this example can be seen as an application of the well-known Laplacian mechanism to the identity query. 
\end{example} 

\begin{example}\label{ex:SetsSan}
Consider again our earlier example where $D = 2^{\mathcal{H}}$ represents the sets of possible hobbies or interests of people.  As noted earlier, it is reasonable to assume that $D$ contains finitely many elements; we denote $|D| = m+1$.  Following Theorem \ref{th:alg} we will construct a mechanism for 1-dimensional databases: this can then be used to define a mechanism for databases in $D^n$ via (\ref{eq:rvdef}).  

For $d\in D$, consider the $D$-valued random variable $X_d$ with probability mass function:
$$\Prob(X_d=d) = 1-pm, \quad \Prob(X_d=d^\prime) = p,$$
where $d\ne d^\prime\in D$.  We make the reasonable assumption that $1-pm > p$.  

For {\ab}-differential privacy, we need the following to hold:
\begin{equation}\label{eq:SetsSan1}
\Prob(X_d\in A)\le e^\epsilon \Prob(X_{d^\prime}\in A)+\delta,
\end{equation}
where $A\subseteq D$ and $d, d^\prime\in D$.

We claim that (\ref{eq:SetsSan1}) will hold if and only if 
\begin{equation}\label{sing}
1-pm\le e^\epsilon p + \delta.
\end{equation}
This condition is clearly necessary as can be seen by considering the singleton set $A = \{d\}$.  To see that it is also sufficient let $d, d^\prime \in D$ and $A\subseteq D$ be given.  There are 4 cases to consider.

\begin{enumerate}
\item $d, d^\prime \notin A$: Then $\Prob(X_d\in A) = \Prob(X_{d^\prime}\in A) = p|A|$ and {\ab}-differential privacy holds trivially.
\item $d, d^\prime \in A$: Then $\Prob(X_d\in A) = \Prob(X_{d^\prime}\in A) = p(|A|-1)+1-pm=p(|A|-m-1)+1$ and {\ab}-differential privacy holds trivially.
\item $d^\prime \in A, d\notin A$: Then $\Prob(X_d\in A) \le \Prob(X_{d^\prime}\in A)$ and {\ab}-differential privacy holds trivially.
\item $d \in A, d^\prime\notin A$: Then \begin{align*}\Prob(X_d\in A) &= p(|A|-m-1)+1\\ \Prob(X_{d^\prime}\in A)&=p|A|.\end{align*}

From (\ref{sing}), we have
\begin{align*}
1-pm&\le e^\epsilon p+\delta\\
&= e^\epsilon (p|A|-p|A|+p)+\delta\\
&\le e^\epsilon p|A| - p(|A|-1)+\delta,
\end{align*}
since $|A|\ge1$ ($d\in A$ by hypothesis). Rearranging the above inequality, we see that 
$$p(|A|-m-1)+1 \le e^\epsilon \left(p|A|\right)+\delta.$$
\end{enumerate}
Thus we can construct an {\ab}-differentially private mechanism of the form (\ref{eq:rvdef}) by choosing $p \geq \frac{1-\delta}{m + e^{\epsilon}}$. 
\end{example}

\section{Accuracy}\label{sec:Acc}
In this section, we consider the question of accuracy for product sanitisations.  The literature on the interaction between privacy and accuracy is considerable and previous work has considered examples such as count queries \cite{DworkAnalysis}, contingency table marginals \cite{BCD07} and spatial data \cite{CPS12}.  As product sanitisations are constructed from 1-dimensional mechanisms, we focus on the error of these basic mechanisms here.  These results can then be used to derive bounds for data in $D^n$; the precise form these bounds will take depends on how the metric is constructed on $D^n$. 

We wish to emphasise two points about our work: we are considering a very general class of databases that can incorporate numerical, categorical and functional data; we consider {\ab}-differential privacy and are not assuming $\delta = 0$.


As $\rho(\cdot, d)$ is a continuous function on $D$ for any fixed $d$, it is measurable with respect to the Borel $\sigma$-algebra on $D$.  It follows that $\rho(X_{d}, d)$ is a nonnegative-valued random variable.  We define the maximal expected error $\E$ as:
\begin{align}\label{eq:error}
\E = \max_{d \in D} \mathbb{E}\left[\rho(X_{d},d)\right].
\end{align}

For $r > 0$ and $x \in D$, $B_r(x)$ denotes the open ball $$B_r(x) := \{y \in D \mid \rho(y, x) < r\}.$$

We first note that for any differentially private mechanism with $\delta < 1$, $\E > 0$.  If $\d = 1$, then any mechanism is differentially private. 
\begin{lemma}
\label{lem:poserr} Let a family $\{X_d: \Omega \to U \mid d \in D\}$ of measurable mappings be given, let $0 \leq \d < 1$ and assume that 
\begin{equation}\label{eq:DP1d}
\mathbb{P}(X_d\in A) \le e^\e \mathbb{P}(X_{d^\prime}\in A)+\d,
\end{equation}
for all $d, d^\prime\in D, A \in \A_U$. Then $\E > 0$.  
\end{lemma}
\begin{proof}
As $D$ is compact, we can choose $u, v$ in $D$ with $\rho(u, v) = \diam(D)$.  Let $r = \frac{\diam(D)}{2}$.  Then from (\ref{eq:DP1d}), it follows that
$$\mathbb{P}(X_u \in B_r(v)) \geq e^{-\e}\left( \mathbb{P}(X_v \in B_r(v)) - \d\right).$$
As $\rho(x, u) \geq r > 0$ for all $x \in B_r(v)$, it follows that $\mathbb{E}\left[\rho(X_u, u)\right] > 0 $ unless 
\begin{equation}
\label{eq:poserr1}\mathbb{P}(X_v \in B_r(v)) = \d.
\end{equation}
However, if this is the case then $\mathbb{P}(\rho(X_v, v) \geq r) = 1-\d >0$ and hence $\mathbb{E}\left[\rho(X_v, v)\right] \geq r(1-\d) > 0$.  This completes the proof.
\end{proof}

We now present two simple results giving lower bounds for $\E$.  The argument for the following result is inspired by that used to establish Theorem~3.3 of \cite{HarTal}.

\begin{theorem}\label{th:AccuracyId} Let a family $\{X_d: \Omega \to U \mid d \in D\}$ of measurable mappings satisfying (\ref{eq:DP1d}) be given. Then 
\begin{equation}\label{eq:accuracy1}
\E \geq (1-\d)\left(\frac{\diam(D)}{2(1+e^{\e})}\right).
\end{equation}
\end{theorem}
\begin{proof}
Without loss of generality, we may assume that $\E$ is finite.  Moreover, from Lemma \ref{lem:poserr} we know that $\E > 0$.  As $D$ is compact, there exist points $u, v$ in $D$ with $\rho(u, v) = \diam(D)$.  Set $t = \frac{\diam(D)}{2\E}$; then $t\E = \frac{\diam(D)}{2}$ and the balls $B_{t\E}(u)$, $B_{t\E}(v)$ are disjoint.  From Markov's inequality applied to the non-negative random variable $\rho(X_u, u)$, it follows that   
\begin{equation}
\label{eq:bd1u}
\mathbb{P}(X_u \in B_{t\E}(u)) \geq 1 - \frac{1}{t} = 1 - \frac{2\E}{\diam(D)}.
\end{equation}
It is now immediate that 
\begin{equation}
\label{eq:bd1uv}
\mathbb{P}(X_u \in B_{t\E}(v)) \leq \frac{2\E}{\diam(D)}.
\end{equation}
From (\ref{eq:DP1d}) we know that 
\begin{equation}
\label{eq:bd2uv} \mathbb{P}(X_u \in B_{t\E}(v)) \geq e^{-\e}(\mathbb{P}(X_v \in B_{t\E}(v)) - \d).
\end{equation}
Combining (\ref{eq:bd1uv}), (\ref{eq:bd2uv}) and noting that (\ref{eq:bd1u}) also holds with $u$ replaced by $v$, we see that
\begin{eqnarray*}
\frac{2\E}{\diam(D)} &\geq& e^{-\e}\left( 1 - \frac{2\E}{\diam(D)} - \d \right)
\end{eqnarray*} 
and after a simple rearrangement of terms we see that
\begin{eqnarray*}
\E \geq (1-\d) \left(\frac{\diam(D)}{2(1+e^{\e})}\right),
\end{eqnarray*}
as claimed.  
\end{proof}

The previous result applies to any compact metric space $D$.  Now assume that $D$ is discrete so that there exists some $\kappa > 0$ such that 
\begin{equation}\label{eq:kappa}
\rho(x, y) \geq \kappa \;\;\quad \forall x, y \in D.
\end{equation}
This combined with $D$ being compact immediately implies that $D$ is finite.  A straightforward alteration of the argument of Theorem \ref{th:AccuracyId} yields the following result.  
\begin{theorem}\label{pro:AccuracyId2} Let $D$ be finite with $|D| = m+1$ and $\kappa = \min_{d, d' \in D} \rho(d, d')$.  Let a family $\{X_d: \Omega \to U \mid d \in D\}$ of measurable mappings be given satisfying (\ref{eq:DP1d}).  Then 
\begin{equation}\label{eq:accuracy1}
\E \geq (1-\d) \left(\frac{\kappa m}{(m+e^{\e})}\right)
\end{equation}
\end{theorem}
\begin{proof}
It is trivial that the $m+1$ balls $B_{t\E}(u)$, $u \in D$ are all disjoint where $t = \frac{\kappa}{\E}$.  Fix some $u \in D$.  By the same reasoning as in the proof of Theorem \ref{th:AccuracyId},
\begin{equation}
\label{eq:discbd1}
\mathbb{P}(X_u \in B_{t\E}(u)) \geq 1 - \frac{\E}{\kappa}.
\end{equation}
Moreover, there must exist some $v \neq u$ such that 
\begin{equation}
\label{eq:discbd2} 
\mathbb{P}(X_u \in B_{t\E}(v)) \leq \frac{\E}{\kappa m}.
\end{equation}
Choose one such $v$ and apply (\ref{eq:DP1d}) to obtain
\begin{equation}
\label{eq:discbd3} 
\mathbb{P}(X_u \in B_{t\E}(v)) \geq e^{-\e}\left( \mathbb{P}(X_v \in B_{t\E}(v)) - \d \right).
\end{equation}
As in the proof of Theorem \ref{th:AccuracyId}, we can now conclude that 
\[\frac{\E}{\kappa m} \geq e^{-\e} \left( 1 - \frac{\E}{\kappa} - \d \right).\]
Rearranging this inequality gives us that 
\[\E \geq (1-\d)\left(\frac{\kappa m}{m + e^{\e}}\right).\]
\end{proof}

\begin{example}
Consider again Example \ref{ex:SetsSan}.  We have shown that there exists an {\ab}-differentially private mechanism with $p= \frac{1-\delta}{m + e^{\epsilon}}$ where $|D| = m+1$.  If $D$ is equipped with the discrete metric so that $\rho(d, d') = 1$ for all $d \neq d'$, then $\kappa =1$ and for any $d$, the expected value of $\rho(X_d, d)$ for this mechanism is given by
$$\sum_{d \neq d'} p = mp = (1-\d) \left(\frac{m}{m + e^{\e}}\right).$$
So the bound given by Theorem \ref{pro:AccuracyId2} is tight in this simple case.  
\end{example}

\section{Concluding Remarks}\label{sec:Conc}
We have considered differential privacy in the setting of probability on metric spaces with mechanisms viewed as measurable functions taking values in output query spaces.  We have demonstrated that this framework allows mechanisms based on sanitisation and output perturbation to be treated in a uniform manner; moreover we have presented examples to highlight that categorical, functional and numerical data can be treated in this setting.  For sanitisation mechanisms, a formal proof that differential privacy with respect to the identity query guarantees differential privacy with respect to any measurable query has been given.  We have also introduced the problem of determining sufficient sets for differential privacy, shown that a generating algebra of sets is a sufficient set and applied this fact to functional data in the space $C([0,1])$ of continuous functions on $[0,1]$.  In the latter half of the paper, we have focussed on product sanitisations of the form (\ref{eq:rvdef}); we have shown that these mechanisms are {\ab}-differentially private if and only if the 1-dimensional mechanism used to define them is.  This result was then applied in two contexts: to provide a condition for the well-known Laplacian mechanism to be differentially private; and to give a simple example of a differentially private mechanism for discrete categorical data.  Finally in Section \ref{sec:Acc}, two simple results giving lower bounds for the maximal expected error for {\ab} differentially private mechanisms on metric spaces were presented.  


\end{document}